\documentclass[final,12pt]{article}
\usepackage{etex}
\usepackage{showkeys}
\usepackage{xspace} 
\usepackage{tikz}   
\usepackage{xypic}
\usepackage{amsmath}
\usepackage{amssymb}
\usepackage{eepic}
\usepackage{graphicx}
\usepackage{epsf}
\usepackage{pstricks}
\xyoption{all}

\usepackage{authblk}

\title{Equations over free inverse monoids with idempotent variables}
\author[1]{Volker Diekert}
\author[2]{Florent Martin}
\author[3]{G\'eraud S\'enizergues}
\author[4]{Pedro V. Silva}

\affil[1]{\small FMI, Universit\"at Stuttgart, Universit\"atsstr. 38, D-70569 Stuttgart, Germany}
\affil[2]{	Fakult\"at f\"ur Mathematik, Universit\"at Regensburg, Universit\"atsstr.~31, 93040 Regensburg, Germany}
\affil[3]{LaBRI, Unit\'e Mixte de Recherche du C.N.R.S. Nr 5800, Universit\'e Bordeaux; 351, cours de la Lib\'eration,	F-33405 Talence Cedex, France}
\affil[4]{Centro de Matem\'{a}tica, Faculdade de Ci\^{e}ncias, Universidade do Porto, R.~Campo Alegre 687, 4169-007 Porto, Portugal}
\newenvironment{proof}{{\bf{Proof.}}}{\hfill$\Box$\\ \medskip}
\newcommand{\qed}{\hfill $\Box$}
\usepackage{prettyref}
\newcommand{\prref}[1]{\prettyref{#1}}
\newrefformat{thm}{Theorem~\ref{#1}}
\newrefformat{lem}{Lemma~\ref{#1}}
\newrefformat{def}{Definition~\ref{#1}}
\newrefformat{cor}{Corollary~\ref{#1}}
\newrefformat{con}{Conjecture~\ref{#1}}
\newrefformat{prop}{Proposition~\ref{#1}}
\newrefformat{sec}{Section~\ref{#1}}
\newrefformat{kap}{Chapter~\ref{#1}}
\newrefformat{ex}{Example~\ref{#1}}
\newrefformat{rem}{Remark~\ref{#1}}
\newrefformat{fig}{Figure~\ref{#1}}
\newrefformat{eq}{Equation~(\ref{#1})}
\newtheorem{theorem}{Theorem}[section]

\newtheorem{lemma}[theorem]{Lemma}
\newtheorem{proposition}[theorem]{Proposition}
\newtheorem{corollary}[theorem]{Corollary}
\newtheorem{definition}[theorem]{Definition}
\newtheorem{example}[theorem]{Example}
\newtheorem{remark}[theorem]{Remark}

\newcommand{\cE}{\mathcal{E}}

\newcommand{\cS}{\mathcal{S}}

\newcommand\last{\mbox{last}}

\def\Box{\square}

\def\tra#1{\smash{\mathop{\mid\kern
-1pt\joinrel\relbar\joinrel\relbar}\limits^{*}_{#1}}}
\def\longtra#1{\smash{\mathop{\mid\kern
-1pt\joinrel\relbar\joinrel\relbar\joinrel\relbar}\limits^{*}_{#1}}}
\def\vlongtra#1{\smash{\mathop{\mid\kern
-1pt\joinrel\relbar\joinrel\relbar\joinrel\relbar\joinrel\relbar}\limits^{*}_{#1}}}
\def\vvlongtra#1{\smash{\mathop{\mid\kern
-1pt\joinrel\relbar\joinrel\relbar\joinrel\relbar\joinrel\relbar\joinrel\relbar}\limits^{*}_{#1}}}
\def\vvvlongtra#1{\smash{\mathop{\mid\kern
-1pt\joinrel\relbar\joinrel\relbar\joinrel\relbar\joinrel\relbar\joinrel\relbar\joinrel\relbar}\limits^{*}_{#1}}}
\def\etra#1{\smash{\mathop{\mid\kern
-1pt\joinrel\relbar\joinrel\relbar}\limits_{#1}}}

\def\X{{\mathcal{X}}}
\def\iff{\Leftrightarrow}

\def\wh{\widehat}

\def\cE{{\mathcal{E}}}

\def\N{\mathbb{N}}

\def\last{\mbox{last}}

\def\max{\mbox{max}}
\def\min{\mbox{min}}

\def\Z{\mathbb{Z}}

\def\bi{\begin{itemize}}
\def\ei{\end{itemize}}
\def\beq{\begin{equation}}
\def\eeq{\end{equation}}

\newtheorem{T}{Theorem}[section]
\newcommand{\bt}{\begin{T}}
\newcommand{\et}{\end{T}}
\newcommand{\ftd}{$\square$\end{T}}

\newtheorem{Proposition}[T]{Proposition}
\newcommand{\bp}{\begin{Proposition}}
\newcommand{\ep}{\end{Proposition}}
\newcommand{\fpd}{$\square$\end{Proposition}}

\newtheorem{Corol}[T]{Corollary}
\newcommand{\bc}{\begin{Corol}}
\newcommand{\ec}{\end{Corol}}
\newcommand{\fcd}{$\square$\end{Corol}}

\newtheorem{Remark}[T]{Remark}
\newcommand{\br}{\begin{Remark}}
\newcommand{\er}{\end{Remark}}
\newcommand{\frd}{$\square$\end{Remark}}

\newtheorem{Example}[T]{Example}
\newcommand{\be}{\begin{Example}}
\newcommand{\ee}{\end{Example}}

\newtheorem{Problem}[T]{Problem}
\newcommand{\bq}{\begin{Problem}}
\newcommand{\eq}{\end{Problem}}

\def\abstract#1{\par\bigskip
\begingroup\small
\baselineskip=12truept
 {\bf Abstract.}  #1}         

\newcommand\pref{\mathop{\mbox{Pref}}}
\newcommand{\dip}[2]{\del_{#2}(#1)}
\newcommand{\FG}{\mathrm{FG}}
\newcommand{\FGRI}[1]{F(#1)}
\newcommand{\FA}{\FGRI{A}}
\newcommand{\FIM}[1]{\mathrm{FIM}(#1)}
\newcommand{\ie}{i.e.,\xspace}
\newcommand{\eg}{e.g.,\xspace}

\newcommand{\IFF}{if and only if\xspace}
\newcommand{\homo}{homomorphism\xspace}
\newcommand{\morph}{morphism\xspace}

\newcommand{\involu}{involution\xspace}
\newcommand{\subst}{substitution\xspace}
\newcommand{\solu}{solution\xspace}
\newcommand{\sateq}{sat-equivalent\xspace}

\newcommand{\lds}{ ,\ldots, }

\newcommand{\set}[2]{\left\{#1\mathrel{\left|\vphantom{#1}\vphantom{#2}\right.}#2\right\}}

\newcommand{\oneset}[1]{\left\{\mathinner{#1}\right\}}
\newcommand{\os}{\oneset}

\newcommand{\es}{\emptyset}
\newcommand{\sse}{\subseteq}

\newcommand{\abs}[1]{\left|\mathinner{#1}\right|}
\newcommand{\Abs}[1]{\left\Vert\mathinner{#1}\right\Vert}

\newcommand{\ceil}[1]{\left\lceil\mathinner{#1} \right\rceil}

\newcommand{\PSPACE}{\ensuremath{\mathsf{PSPACE}}}
\newcommand{\DEXPTIME}{\ensuremath{\mathsf{DEXPTIME}}}

\newcommand{\NP}{\ensuremath{\mathsf{NP}}}

\renewcommand{\phi}{\varphi}

\newcommand{\alp}{\alpha}
\newcommand{\bet}{\beta}
\newcommand{\gam}{\gamma}
\newcommand{\del}{\delta}

\newcommand{\sig}{\sigma}
\newcommand{\Sig}{\Sigma}
\newcommand{\Gam}{\Gamma}

\newcommand\OO{\Omega}
\newcommand\XX{\X}
\newcommand{\Oh}{\mathcal{O}}

\newcommand{\ov}[1]{\overline{#1}}
\newcommand{\oi}[1]{{#1}^{-1}}
\newcommand{\invol}{\overline{\,^{\,^{\,}}}}

\date{\today}

\begin{document}

\maketitle

\begin{center}\small
2010 Mathematics Subject Classification: 20M18, 20F70, 03D40

\bigskip

Keywords: equation, language equation, free inverse monoid, idempotent variable,
one-variable equation.
\end{center}
\abstract We introduce the notion of idempotent variables for studying
equations in inverse monoids. It is proved that it is decidable in
singly exponential time (DEXPTIME) whether a system of equations in
idempotent variables over a free inverse monoid has a solution. 
Moreover the problem becomes hard for DEXPTIME, as soon as the quotient group of the free inverse monoid has rank at least two. 
The upper bound is proved by a direct reduction to solve language equations
with one-sided concatenation and a known complexity result by Baader
and Narendran (\emph{Unification of concept terms in description logics},
2001). For the lower bound we show hardness for a restricted class of language equations. 

Decidability for systems of typed equations over a free inverse monoid
with one irreducible variable and at least one unbalanced equation is
proved with the same complexity upper-bound.
 
Our results improve known complexity bounds by 
Deis, Meakin, and S{\'e}nizergues (\emph{Equations in free inverse monoids}, 2007).
 Our results
also apply to larger families of equations where no decidability has
been previously known. The lower bound confirms a conjecture made in the 
conference version of this paper which was presented at Computer Science in Russia (CSR 2015).

\section{Introduction}
It is decidable whether 
equations over free monoids and free groups are solvable. 
These classical results were proved by Makanin in his seminal papers \cite{mak77,mak82}. 
A first estimation of the time complexity for deciding solvability was more than triple or four times exponential, but over the years it was lowered. It went down to PSPACE by Plandowski \cite{pla99focs,pla04jacm} for free monoids. Extending his method Guti{\'e}rrez showed that the same complexity bound 
applies in the setting of free groups \cite{gut2000stoc}. In \cite{Jez15jacm} Je\.z  used his ``recompression technique'' and achieved the best known space complexity to date: NSPACE($n \log n$). Perhaps even more importantly, 
he presented the simplest known proof deciding the problem {\sc Wordequations}
leading to an easy-under-stand algorithmic description for the set 
of all solutions for equations over free monoids and free groups (with
rational constraints) \cite{DiekertJP2014csr}. Actually, \cite{CiobanuDEicalp2015} showed that the set of all solutions 
in reduced words 
over a free group is an indexed language. More precisely it is an EDT0L language.

In the present paper we study equations over inverse monoids. Inverse monoids are monoids with involution and constitute the most natural
intermediate structure between monoids and groups. They are well-studied and pop-up in various applications, for example when investigating systems which are deterministic and codeterministic. Inverse monoids arise naturally as monoids of injective
transformations closed under inversion. Indeed, up to isomorphism,
these are all the inverse monoids, as stated in the classical
Vagner-Preston representation theorem. This makes inverse monoids
ubiquituous in geometry, topology and other fields.


The fifties of the last century boosted the systematic study of inverse monoids. However, the word problem remained unsolved until the early seventies, when
Scheiblich \cite{Scheiblich73} and Munn \cite{Munn:74} independently provided
solutions for free inverse
monoids. The next natural step is to consider solvability of equations, \ie the existential theory. Rozenblat's paper \cite{roz85}
destroyed all hope for a general solution: solving equations in free inverse
monoids is undecidable. Thus,
the best we can hope is to prove decidability for particular
subclasses. For almost a decade, the
reference paper on this subject has been the paper of Deis, Meakin, and
S\'enizergues \cite{DMS07ijac}. The authors considered the following 
\emph{lifting problem}. The input is given by an equation over a free inverse monoid together with a solution over the free quotient. The question is whether the solution over the group can be lifted to a solution in the inverse monoid. 
\cite{DMS07ijac} showed decidability of the lifting problem using Rabin's tree theorem. The result is an algorithm which is super-exponential (and 
at least doubly exponential in their specific setting). 

In the present paper, we achieve various improvements. Our main result lowers the complexity of the lifting problem to singly exponential time; and as soon as the input is a system of at least two equations, then the the lifting problem becomes $\DEXPTIME$-hard. Moreover, we study equations with idempotent variables instead of lifting properties, which leads to a uniform approach and a simplified proof. It also enabled us to generalize some results concerning one-variable equations to a broader setting, thereby
leading to new decidability results. 

A more precise statement about the progress achieved is as follows. First, Theorem~\ref{thm:marie} shows that deciding solvability of systems of equations
in idempotent variables over $\FIM{A}$ is $\DEXPTIME$-complete. The upper bound improves the \cite[Thm.~8]{DMS07ijac}. Our proof is based on a  well-known result {}from \cite{BaaderN01} by Baader and Narendran, while
the complexity of the algorithm in \cite[Thm.~8]{DMS07ijac} is
much higher, since the algorithm involves Rabin's Tree Theorem\footnote{
The $\DEXPTIME$ upper bound was obtained first by 
the second and third author, but not published. The same improvement was discovered later independently by the two other authors; and the present paper joins both approaches. In addition, we take the opportunity to correct a mistake in \cite{DMS07ijac} about some special one-variable equations, where Assumption~2 in \prref{def:numberedsunny} was missing.
}. The lower bound, which is $\DEXPTIME$-hardness (for systems of two equations
and where the quotient group of the free inverse monoid has rank at least two) confirms a conjecture in the conference version of the present 
paper \cite{DiekertMSScsr15}. 
Second, with respect to unbalanced one-variable equations
and \cite[Thm.~13]{DMS07ijac},
our Theorem~\ref{thm:frida} admits the presence of arbitrarily many
idempotent variables, and the complexity
very much improved in view of Theorem~\ref{thm:marie}.
Morover, our proofs are shorter and easier to understand by a direct reduction to 
language equations. 

\section{Preliminaries and notation}\label{sec:prel}
\paragraph{Sets and finite subsets.}Given a set $S$, we denote by $2_f^S$ the set of {\em finite} subsets of the set $S$.
\paragraph{Complexity.}
A function $p: \N \to \N$ is called {\em polynomial} if $p(n) \in n^{\Oh(1)}$. 
It is  called \emph{singly exponential} if $f(n) \leq  2^{p(n)}$ where $p$ is some
polynomial. 
The complexity class 
$\DEXPTIME$ refers to problems which 
can be solved 
on deterministic Turing machines within a  
singly 
exponential time bound. A problem is encoded as a subset over the binary alphabet $\os{0,1}$.
A problem $P$ is $\DEXPTIME$-hard, if for every problem $L \in \DEXPTIME$
there exists a polynomial time computable function $f:\os{0,1}^* \to \os{0,1}^*$ such that: $w\in L \iff f(w) \in P$. It is called $\DEXPTIME$-complete if it belongs to $\DEXPTIME$ and, in addition, it is $\DEXPTIME$-hard. In a few places we also refer to other complexity classes
like $\PSPACE$ (=\emph{polynomial space}) or $\NP$ (=\emph{nondeterministic polynomial time}).
The notation is standard, see for example \cite{pap94}. As usual in the literature, explicit encodings of problems are omitted. Our reductions are actually ``logspace'' reductions. Formally, this makes the lower bound results stronger, but this is not our primary goal: so we stay within the framework of polynomial-time reductions. 

\paragraph{Monoids and groups.}
A \emph{monoid} is a nonempty set $M$ with a binary associative operation: $(x,y) \mapsto x\cdot y$ together with a neutral element $1$ satisfying $1\cdot x=x\cdot 1=x$ for all $x \in M$. Frequently, we write $xy$ instead of $x \cdot y$. A group is a monoid $G$ where for each $x\in G$ there exists some $\ov x \in G$ such that $x \ov x = 1$. If $G$ is a group, then its \emph{inverse} $\ov x= \oi x$ is uniquely defined.

\paragraph{Words and languages.}
An {\em alphabet} is a (finite) set; and an element of an alphabet is 
called a \emph{letter}. The free monoid generated by an alphabet
$A$ is denoted by $A^*$.  
The elements of $A^*$ are called \emph{words}: these are the finite sequences of letters. The empty word is denoted
by $1$ as the neutral element in other monoids as well, provided the operation is written as a multiplication.  
The \emph{length} of a word $u$ is denoted by $\abs{u}$. We
have $\abs{u} = n$ for $u= a_1 \cdots a_n$ where $a_i \in A$. The
empty word has length $0$, and it is the only word with this property.
A word $u$ is a \emph{factor} of a word $v$ if there exist $p,q \in A^*$ such
that $puq = v$. 
It is a \emph{prefix}, if $uq = v$ for some $q \in
A^*$, and it is a \emph{suffix}, if $pu = v$ for some $p \in A^*$.  
A \emph{language} $L$ over $A$ is a subset of $A^*$. It is called 
\emph{factor-} (resp.~\emph{prefix-}) (resp.~\emph{suffix-}) \emph{closed}, 
if with every $u$ every {factor}, (resp.~{prefix}), (resp.~{suffix}) of $u$ belongs to $L$ as well. We write  $\pref(L)$ for its prefix-closure, thus 
$$\pref(L)= \set{u\in A^*}{\exists v \in L: u \leq v}.$$

\paragraph{Involutions.}
 An \emph{involution}  
is a mapping $\invol$ such that 
$\overline{\overline{x}} = x$ for all elements. 
In particular, an involution is a bijection.  The identity is an \involu.

\paragraph{Monoids with involutions.}
If an involution is defined for a monoid, then we additionally require $\overline{xy}=\overline{y}\,\overline{x}$ for all its elements $x, y$. Every group is a monoid with involution by letting 
$\ov x = \oi x$. 

If an alphabet is equipped with an involution, then we extend it to the free monoid $A^*$ by 
$$\ov{a_1 \cdots a_m}= \ov{a_m} \cdots \ov{a_1}.$$
When $\ov a = a$ for all $a \in A$, then $\ov{w}$ simply means to read the word from right-to-left. Every alphabet $B$ (without involution) can be embedded into an alphabet $A$ with involution without fixed points by letting
$A = B\cup \ov B$ where $\ov B= \set{\ov a}{a \in B}$ is a disjoint copy of 
$B$. The involution maps $a\in B$ to $\ov a$ and vice versa.

The identity is a \morph for monoids with \involu \IFF the monoid is commutative. 
In particular, given a set $S$ the set of finite subsets $2_f^S$ is a commutative  
monoid where the operation is the union. Thus, we also write 
$L+K$ instead of $L \cup K$. Elements of $s \in S$ are identified with singletons $\os s \in 2_f^S$. According to the additive notation the neutral element in $2_f^S$ is denoted as $0$. We have $0 = \es$. Actually, in our application we have 
$1 \in S \sse A^*$ and then $1\in 2_f^S$ denotes the singleton $\os 1$. 
Thus, $0 \neq 1$ in $2_f^S$. There is however no risk of confusion: similar conventions are standard for $\N$ or $\Z$. 

\paragraph{Homorphisms and~\morph{s}.}
A \emph{\homo} is a mapping which respects the algebraic structure, whereas the notion {\em \morph} refers to a mapping which 
respects the involution and in addition, 
depending on the category, respects the algebraic structure, too.
Hence, a  \morph between sets with involution is just a mapping respecting 
the involution, whereas a \morph between monoids with involution
is a monoid \homo respecting the involution. A \homo between groups is a \morph. 

\paragraph{Free groups.}
Let $A$ be an alphabet with \involu. It defines a quotient group $\FA$ by adding defining relations $a \ov a= 1$ for all $a \in A$. 
If we can write $A$ as a disjoint union $B \cup \set{\ov a}{a \in B}$, then 
$\FA$ is nothing but the free group $\FG(B)$ in the standard meaning. In general, $\FA$ 
is a free product of a free group with cyclic groups of order $2$. Although our primary interest is the usual free group $\FG(B)$, the notation $\FA$ is more convenient for us. Moreover, various results hold for $\FA$ and without changing the proofs. Last but not least, $\FA$ is the ``free group'' with respect to the category of sets with involution: every \morph of $A$ to a group $G$ extends uniquely to a \morph from $\FA$ to $G$.

As a set (with involution) we can identify $\FA$ with the subset of reduced words in $A^*$. As usual, a word is called \emph{reduced} if it does not contain any factor $a \ov a$ where $a \in A$. Observe that this embedding 
of $\FGRI A$ into $A^*$ is indeed compatible with the involution. 
In the following we let $\pi: A^* \to \FGRI A$ be the canonical
\morph from $A^*$ onto $\FA$.
It is well-known (and easy to see) that
every word $u \in {A}^*$ can be transformed into a unique reduced
word $\wh{u}$ by successively erasing factors of the form $a\ov{a}$ where $a \in A$. This leads to the assertion
$$\forall u,v \in {A}^*: \pi (u)= \pi(v) \; \iff \; \wh{u} =
\wh{v}.$$

We systematically identify the set 
$\FA$ with the subset  of reduced words in $A^*$. Concepts such as length, factor, 
prefix, and prefix-closure are inherited from free
monoids to free groups  via reduced words. 
For the same reason, it makes sense to write $ \wh u = \pi (u),$ 
for $u\in A^*$, because $\pi(u) \in \FA$
is identified with $\wh u \in A^ *$. If 
 $L\sse A^*$ is prefix-closed, then $\wh L= \set{\wh u}{u \in L}\sse \FA$ is prefix-closed, too (\prref{lem:pretopre}). We have 
$\wh L \sse \wh{{A}^*}= \FA \sse A^*$.

\paragraph{Free inverse monoids}
A monoid $M$ is said to be {\em inverse} if for every $x\in M$ there
exists a unique element $\ov x\in M$ satisfying  $x \ov x x = x$
and $\ov x x \ov x= \ov x$. Clearly, $\ov {\ov x}= x$ by uniqueness of $\ov x$ and, hence,  $M$ is a set with involution. The mapping 
$x \mapsto \ov x$ is also called  an \emph{inversion}.
Idempotents commute in inverse monoids (see \eg \cite{pet84}), hence the subset 
$E(M) = \set{ e \in M}{ e^2 = e}$ is a commutative submonoid. 
Since necessarily $\ov e = e$ for $e \in E(M)$ one easily deduces that 
$\ov {xy}= \ov y \, \ov x$ for all $x,y\in M$. As a consequence, 
an inverse monoid is a monoid with involution. 

In the literature the notation $\ov x = \oi x$ is also used for elements of inverse monoids, just as for groups
(which constitute a proper subclass of inverse monoids). By default, 
the involution on an inverse monoid (and hence in every group) is supposed to be given by its inversion.
We proceed now to describe Scheiblich's construction of the free inverse
monoids $\FIM{A}$ where $A$ is an alphabet with \involu.  

The  elements of $\FIM{A}$ are pairs 
$({P},g)$, where the second component is a group element $g \in \FA$ and the first component is a finite prefix-closed subset ${P}$ 
of $\FA$ such that $g \in {P}$. In other terms, this means that ${P}$ is a finite connected subset of the Cayley graph of $\FA$ (over $A$) such that $1, g \in {P}$. Formally, we let 
$$\FIM{A} = \set{ ({P},g) }{\abs {P} < \infty  \wedge  g \in {P}= \pref({P}) \sse \FA}.$$
The multiplication  on ${\FIM{A}}$ is defined through
$$({P},g)({Q},h) = ({P} \cup g{Q},gh).$$
It is easy to see that ${\FIM{A}}$ is a monoid with identity $(\os{ 1 },1)$
and every $({P},g)$ has a unique inverse $(\oi g {P},\oi g)$, hence ${\FIM{A}}$ is an inverse monoid.

Let $\psi:{A}^* \to {\FIM{A}}$ be the homomorphism of monoids defined by
$\psi(a) = (\{ 1,a\},a).$ Then we have $\psi(\ov a) = (\os{ 1,\ov a},\ov a) = \ov{(\{ 1,a\},a)}$ 
and  $\psi$ is a \morph of monoids with involution. 
We obtain the universal property of being free with respect
to sets with involution: 
let $M$ be an inverse monoid  and  
$\phi: A \to M$ a \morph of sets with involution, then there is exactly one \morph $\eta: \FIM A \to M$
of monoids with involution such that $\Phi(a) = \phi(a)$ for all 
$a \in A$. 
In other words, let $\iota = \psi|_A$ and 
$\phi:A \to M$ be any mapping respecting the involution where $M$ is an inverse monoid. Then there
exists a unique \morph of inverse monoids $\eta:{\FIM{A}} \to M$ such
that the following diagram commutes.
$$\xymatrix{
A \ar[d]^{\phi} \ar[rr]^{\iota} && \FIM{A} \ar@{-->}[dll]^{\eta} \\
M & \  
}$$ In particular, 
$\pi:   A^* \to \FGRI A$ factorizes through~$\eta$. The monoid ${\FIM{A}}$ is, up to isomorphism, uniquely defined by this universal property: ${\FIM{A}}$ is a {\em free
inverse monoid} in the category of sets with involution.
If $A$ can be written as a disjoint union $A = B \cup \set{\ov a}{a \in B}$, then ${\FIM{A}}$ is the {\em free
inverse monoid} over $B$ in the category of sets (without involution).

The following  diagram summarizes our notation. 
$$\xymatrix{
A^* \ar[d]^{\pi} \ar[r]^{\!\psi} & \FIM{A} \ar@{-->}[dl]^{\eta} \\
\FGRI A  &    \!\!\!\!    \!\!\!\! = \quad \set{\wh w}{w \in A^*} \sse  A^*& \hspace{-1cm}
\text{ as sets}
}$$

\section{Language equations}\label{sec:lefm}
Henceforth, $A$, and $B$ denote  alphabets of constants and $\OO$ denotes 
an alphabet of variables. The alphabets are finite and disjoint. We assume that $B \sse A$ and that $A\cup \OO$ is a set with involution. 
However, for technical reasons we require that 
$X= \ov  X$ for all variables. 
We use $a, b, c, \ldots$ to denote letters of $A$, whereas variables are denoted by capital letters $X,Y,Z \ldots$.

Our complexity results for solving certain equations over free inverse
monoids rely on a paper of Baader and Narendran \cite{BaaderN01}.
The paper shows that the satisfiability problem of language equations with one-sided concatenation is $\DEXPTIME$-complete 
for free monoids. As we need the corresponding result for free groups as well, we define the notion of language equation and its solutions in a more general framework.

A \emph{system of language equations} $\cS$ (with one-sided concatenation) has the form 
\begin{equation}\label{eq:ggs0}
L_{k} + \sum_{i\in I_k} u_{ki}X_i = K_k + \sum_{j\in J_k} u_{kj}X_j \quad \text{ for } 1\leq k \leq n.
\end{equation}
Here, $n \in \N$ and $I_k$, $J_k$ are finite (disjoint) index sets,  $L_k$, $K_k$ are
finite subsets of $A^*$, $u_{ki}, u_{kj} \in A^*$ are words, 
and $X_{i}, X_{j} \in \OO$. 

If $L_k$, $K_k$ are
 subsets of $B^*$ and  $u_{ki}, u_{kj} \in B^*$, then we say that $\cS$ is a system with coefficients over $B$. 

The \emph{size} of $\cS$ is defined as 
$$\Abs \cS = \abs{A\cup \OO}+ \sum_{k= 1}^n \left(\abs{I_k} + \abs{ J_k} + \sum_{u\in L_k \cup K_k}\abs {u} + \sum_{i \in I_k} \abs{u_{ki}} + \sum_{j\in J_k} \abs {u_{kj}}\right).$$

\begin{example}\label{ex:ggs1}
Recall that a word  $u$ is identified with the 
singleton $\os u \sse A^*$. Consider $A= \os{a, \ov a, b, \ov b}$ and \begin{equation}\label{eq:ggs1}
a \ov a+ a\ov aX + b\ov bY = b\ov b + a\ov aY + b\ov bX . 
\end{equation}
It is a system in one equation and its size is $22$.
\end{example}

The notion of solution depends on the context. In our paper we use solutions in 
finite subsets of free groups and free monoids. Let $M$ denote  either the free monoid 
$A^*$ or the group $\FGRI A$. In particular, we have inclusions of sets with \involu 
$A \sse M \sse A^*$; and $A$ generates $M$ as a monoid. 
%

A \emph{solution} of $\cS$ in (\ref{eq:ggs0}) is a mapping $\sig: \OO \to 2^{A^*}_f$ 
such that 
\begin{equation*}
L_{k} + \sum_{i\in I_k} u_{ki}\sig(X_i) = K_k + \sum_{j\in J_k} u_{kj}\sig(X_j)
\end{equation*}
becomes an identity in $2_f^M$ for all $1 \leq k \leq n$.
Thus, a solution 
substitutes each $X\in \OO$ by some finite 
subset $\sig(X)$ of $A^*$, but the interpretation is in $M$. 
Of course,  for $M=\FA$ we can demand that each $\sig(X)$ must be a finite subset in reduced words:  $\sig(X)\subseteq \FA$.

\begin{theorem}[\cite{BaaderN01}, Thm.~6.1 and Thm.~7.6]\label{thm:bn01}
The following problem can be solved in $\DEXPTIME$; and it is $\DEXPTIME$-complete for  $\abs B \geq 2$.

{\bf Input.} A system $\cS$ of language equations with coefficients over $B$.

{\bf Question.}  Does $\cS$ have a solution in the free monoid $B^*$?
\end{theorem}

\begin{remark}\label{rem:bn7.6}
\cite{BaaderN01} states \prref{thm:bn01}
for a single equation. However, this covers the general case. 
Indeed, assume that a system $\cS$ 
of  language equations over the free monoid 
$A^*$ has $n$ equations. Without restriction we have $\abs A \geq 2$.  Choose $n$ pairwise different words $p_1, \ldots, p_n\in A^*$ of equal length (say $\ceil{\log_2 n}$);  and for $1 \leq k \leq n$ replace 
the $k$-th equation 
$L_{k} + \sum_{i\in I_k} u_{ki}X_i= K_k + \sum_{j\in J_k} u_{kj}X_j$ 
by 
$$p_kL_{k} + \sum_{i\in I_k} p_k u_{ki}X_i = p_k K_k +  \sum_{j\in J_k} p_k u_{kj}X_j.$$ 
 Summing all left-hand sides and all right-hand sides yields a single equation
\begin{equation}\label{eq:sinLeq}
 \sum_{k=1}^n (p_kL_{k} + \sum_{i\in I_k} p_k u_{ki}X_i) = \sum_{k=1}^n (p_k K_k +  \sum_{j\in J_k} p_k u_{kj}X_j).
\end{equation}
The reduction works since $\os{p_1 , \ldots, p_n}$ is a prefix code. 
Note that the transformation of the system $\cS$ 
to \prref{eq:sinLeq} preserves the set of \solu{}s.
\end{remark}
Consider again \prref{eq:ggs1}: 
$a \ov a+ a\ov aX + b\ov bY = b\ov b + a\ov aY + b\ov bX$. 
 Over $M = \FA$ the equation becomes trivial: 
it states $1 + X+Y= 1 + Y + X$ which is a tautology. Hence every \subst 
in finite subsets of $A^*$ is a solution over $M$. However, for $M = A^*$ the 
structure of solutions is more restricted. 
In the spirit of  \prref{rem:bn7.6} we see that \prref{eq:ggs1} encodes over $A^*$
a system of two equations: $1+X = Y$ and $Y = 1+X$. 
Hence, the set of solutions over $A^*$ is the 
set of mappings $\sig: \OO \to 2^{A^*}_f$ 
such that $\sig(X) =\sig(Y)$ and $1 \in \sig(X) \cap \sig(Y)$. 
\section{Typed equations over free inverse monoids}\label{sec:typeq}
 An {\em equation} over $\FIM{A}$ is a pair $(U,V)$
of words over $A\cup\XX$, sometimes written as $U=V$. Here $A$ is an alphabet of 
\emph{constants} and $\XX$ is a set of variables. Variables $X\in \XX$ represent elements in $\FIM{A}$ and therefore $\XX$ is an  alphabet with involution, too.
Without restriction we may assume $X \neq \ov X$ for all $X \in \XX$.
A {\em solution} $\sig$
of  $U=V$ is a mapping $\sig: \XX \to A^*$  such that $\sig(\ov X) = \ov{\sig(X)}$ for all $X \in \XX$ and such that the replacement of variables by the substituted words in $U$ and in $V$ give the same element in $\FIM{A}$, \ie $\psi(\sig(U)) = \psi(\sig(V))$ in $\FIM{A}$, where $\sig$ is extended to a \morph $\sig: (A\cup \XX)^* \to A^*$
leaving the constants invariant. Clearly, we may specify  $\sig$ also by a mapping {}from $\XX$ to $\FIM{A}$.
For the following it is convenient to have two more types of variables which are used to represent specific elements in $\FIM{A}$. We let $\OO$ be a set of \emph{idempotent variables} and  $\Gam$ be a set of \emph{reduced variables}. Both sets are endowed with an involution. 
We let $\ov Z = Z$ for idempotent variables and $\ov x \neq x $ for all reduced variables. Thus, idempotent variables 
are the only variables which are self-involuting; and variables in $\Gam$ or $\XX$ are not self-involuting. 
We also insist that $A$, $\XX$, $\OO$, and $\Gam$ are pairwise disjoint. 
A {\em typed equation} over $\FIM{A}$ is a pair $(U,V)$ of 
words over $A\cup\OO \cup \Gam$.  A {\em system of typed equation}
is a collection $\cS$ of typed equations; and 
a {\em solution} $\sig$
of $\cS$ is given by a mapping respecting the involution
{}from $\OO \cup \Gam$ to $A^*$, 
which is extended to a \morph $\sig:(A \cup \OO \cup \Gam)^*\to A^*$
respecting the involution and letting the letters of $A$ invariant, 
such that the following conditions hold. 
\begin{enumerate}
\item $\psi(\sig(Z))$ is idempotent for all $Z \in \OO$.
\item $\sig(x)$ is a reduced word for all $x\in \Gam$.
\item We have $\psi(\sig(U))=\psi(\sig(V))$ for all $(U,V) \in \cS$. 
\end{enumerate}

\begin{lemma}\label{lem:taueqtoreq}
Let $(U,V)$ be an (untyped) equation over $\FIM{A}$. For each $X, \ov X\in \XX$ 
choose a fresh idempotent variable $Z_X\in \OO$ and fresh reduced 
variables  $x_X, \ov{x}_{X}\in \Gam$. 
Let $\tau$ be the word-substitution (i.e. monoid homomorphism) 
which replaces each $X, \ov X \in \OO$
by $Z_Xx_X$ and $\ov{x}_XZ_X$ respectively. 
If $\sig$ is a solution of  $(U,V)$ then  a solution $\sig'$ for $(\tau(U),\tau(V))$ 
can be defined as follows. For $\sig(X)= (P,g)$, where $g$ is represented by a 
reduced word, we let  $\sig'(Z_X) = (P,1)$ and $\sig'(x_X) = (\pref(g),g)$. \\
Conversely, if
 $\sig'$ solves $(\tau(U),\tau(V))$ with $\sig'(Z_X) = (P,1)$ and 
$\sig'(x_X) = (\pref(g),g)$ then  $\sig(X)= (P \cup \pref(g),g)$ defines a solution 
for $(U,V)$.
\end{lemma}
\begin{proof}
Trivial. 
\end{proof}
%

By \prref{lem:taueqtoreq} we can reduce the satisfiability of 
equations in $\FIM{A}$ to  satisfiability of typed equations. The framework of typed equations is  more general; and it fits better to our formalism. 
Let  $(U,V)$ be a typed equation, by the \emph{underlying group equation}
we mean the pair $(\pi(U),\pi(V))$ which is obtained by erasing all idempotent variables. Clearly, if $(U,V)$ is satisfiable then 
$(\pi(U),\pi(V))$ must be solvable in the free group $\FGRI A$.
This leads to the idea of \emph{lifting} a solution of a group equation 
to a solution of $(U,V)$ in $\FIM{A}$. It has been known by \cite{DMS07ijac} that
it is decidable whether a lifting is possible. The following result
improves decidability by giving a deterministic exponential time bound. 

\begin{theorem}\label{thm:mainlift}
The following problem is in $\DEXPTIME$. 

{\bf Input.} A system $\cS$ of equations over 
 $\FIM{A}$ and a solution $\gam: \Gam \to \FGRI A$ of the system $\pi(\cS)$  of underlying group equations. 
 
{\bf Question.} Does $\cS$  have a  solution $\sig: \XX \to \FIM A$
 such that $\gam= \eta \circ \sig$? 
 
 Moreover, the problem becomes $\DEXPTIME$-hard as soon as 
$A$ contains four pairwise different letters $\alp, \ov \alp, \bet$, and $\ov \bet$, the system has two equations, and $\gam$ is the trivial mapping 
$\gam(x) = 1$ for all $x \in \Gam$. 
\end{theorem}

\begin{proof}
 For the upper bound we proceed as follows. 
 Due to \prref{lem:taueqtoreq} we first transform the system into a 
 new system with variables in $\OO \cup \Gam$. 
 Next we replace every reduced variable 
$x \in \Gam$ by $(\pref(\sig'(x)),\sig'(x))$. Since the solution 
is part of the input this increases the size of $\cS$ at most quadratic. We obtain a system of equations in idempotent variables and we apply \prref{thm:marie} in 
\prref{sec:slefg} below. 
Actually, \prref{thm:marie} shows also the lower bound, because fixing 
$\gam: \Gam \to \FA$ to be the trivial mapping  means that every lifting 
$\sig: \XX \to \FIM A$ turns $\sig(X)$ into an idempotent. Thus, fixing 
$\gam: \Gam \to \FA$ to be the trivial mapping, leads directly to the framework of idempotent variables. 
 \end{proof}
 
 The next result combines \prref{thm:mainlift} and a known
 complexity result for systems of equations over free groups \cite{DiekertJP2014csr}. 
 
\begin{corollary}\label{cor:mainsota}
Let $\cS$ be a system of equations over the free inverse monoid 
 $\FIM{A}$ and $\pi(\cS)$ the system of underlying group equations. 
 \begin{enumerate}
\item\label{sotai} On input $\cS$ it can be decided in polynomial space whether the system 
$\pi(\cS)$ of group equations has at most finitely many solutions. 
If so, then every solution has at most doubly exponential length. 
\item\label{sotaii} On input $\cS$ and the promise that $\pi(\cS)$ has at most finitely many solutions it can be decided in deterministic triple exponential time whether $\cS$ has a solution. 
\end{enumerate}
\end{corollary}

\begin{proof}
The statement \ref{sotai} follows {} from \cite{DiekertJP2014csr}.
In particular, the size of the set of all solutions is at most triple exponential. Since the square of a triple exponential function is 
triple exponential again, 
the statement \ref{sotaii} follows  from \prref{thm:mainlift}. 
\end{proof}

\section{Solving equations in idempotent variables}
\label{sec:slefg}
\prref{thm:marie} is the main result of the paper.
We split its proof into two sections. \prref{sec:ub} shows the membership 
to $\DEXPTIME$. \prref{sec:lbmarie} shows $\DEXPTIME$-hardness for systems with two  equations.  \prref{thm:marie} improves via \prref{thm:mainlift} the result 
  \cite[Thm.~8]{DMS07ijac}, which was derived from Rabin's Tree Theorem leading to a super-exponential complexity. It improves the main result of \cite{DiekertMSScsr15} since it also shows the conjecture that solving equations in idempotent variables is $\DEXPTIME$-complete.

\begin{theorem}\label{thm:marie}
The following problem can be decided in $\DEXPTIME$.

{\bf Input.} A system $\cS$ of equations in idempotent variables
(\ie without any reduced variable). 

{\bf Question.} Does $\cS$ have a solution in $\FIM A$?

Moreover, if 
$A$ contains four  pairwise different letters $\alp, \ov \alp, \bet, \ov \bet$, then the problem is $\DEXPTIME$-hard for systems with two equations. 
\end{theorem}

\subsection{Upper-bound: containment in $\DEXPTIME$}\label{sec:ub}
This section proves the upper bound mentioned in \prref{thm:marie}. 
We begin with the following lemma. 
\begin{lemma}\label{lem:from_idempotents_to_subsets}

There is a polynomial time algorithm for the following computation. 

{\bf Input.} A finite alphabet with involution $A$ and an equation \begin{equation}\label{eq:start}
u_0 X_1 u_1  \cdots X_{g}u_g = 
v_{0} Y_{1} v_1 \cdots Y_{d}v_d, 
\end{equation}
where the $X_i$'s and $Y_j$'s are 
idempotent variables and such that  the identity $u_0 \cdots u_g = 
v_0 \cdots v_d$ holds in the group $\FA$. 

{\bf Output.} A language equation  which is solvable in nonempty, finite, prefix-closed subsets of $\FA$ \IFF \prref{eq:start} is solvable in $\FIM{A}$.
\end{lemma}
\begin{proof} Define for $0\leq i \leq g$ and $0\leq j \leq d$ the words 
$$p_i=u_0 \cdots u_i,\;\;q_j=v_0 \cdots v_j \in A^*.$$
In every inverse monoid we have $pZ=pZ\ov{p}p$ for every idempotent $Z$ and every element $p$. 
Since $p_{i-1} u_i = p_i$ and $q_{j-1}v_j = q_j$, the equation in idempotent variables
(\ref{eq:start}) can be rewritten as:
\begin{multline}
\label{eq:conjugates}
p_0 X_1 \ov{p_0} \cdots p_iX_{i+1}\ov{p_i} \cdots p_{g-1}X_g\ov{p_{g-1}} \cdot p_g \ov{p_{g}}p_g= \\
q_0 Y_1 \ov{q_0}\cdots q_jY_{j+1}\ov{q_j} \cdots q_{d-1}Y_d\ov{q_{d-1}} \cdot q_{d}\ov{q_{d}}q_d.
\end{multline}
By hypotheses we have $p_g= q_d$ in $\FA$. Moreover, idempotents commute. Hence, 
\prref{eq:conjugates} is equivalent in $\FIM{A}$ with the following equation
\begin{equation}
\label{eq:idempotents_and_freeg}
p_g \ov{p_{g}} \cdot\prod_{i=0}^{g-1}p_i X_{i+1}\, \ov{p_i} = q_d \ov{q_{d}}
\cdot\prod_{j=0}^{d-1}q_j Y_{j+1} \,\ov{q_j}.\end{equation}
Each value of $X_i$, resp.~$Y_j$, in $\FIM{A}$ has the form $(P_i,{1})$, resp.~$(Q_j,{1})$, for 
non-empty, prefix-closed, and finite subsets $P_i$ and $Q_j$  of $\FA$.

Recall that $\wh u$ refers to the reduced word $\pi(u) \in \FA\sse A^*$. 
Define $L= \set{\wh p}{p\in \pref(p_g)}$ and $K= \set{\wh q}{q\in \pref(q_d)}$.
Then the output of the algorithm is the language equation where the solutions $X_i$'s and $Y_j$'s are required to be nonempty, finite, prefix-closed subsets of $\FA$: 
\begin{equation}
\label{eq:ssfrg}
L + \sum_{i=0}^{g-1}\wh{p_i} X_{i+1} = 
K + \sum_{j=0}^{d-1}\wh{q_j} Y_{j+1}.
\end{equation}
It follows from the construction and Scheiblich' s presentation of free inverse monoids that $\sig(X_i)= (P_i,{1})$ and $\sig(Y_j)= (Q_j,{1})$
solves \prref{eq:start} in $\FIM{A}$ \IFF 
$\sig'(X_i)= P_i$ and $\sig'(Y_j)= Q_j$ solves \prref{eq:ssfrg} in $\FA$. 
Hence, the lemma.
 \end{proof}
 
We also make use of the following easy observation.
\begin{lemma}\label{lem:pretopre}
 Let $P\sse A^*$ be  prefix-closed and 
 $\wh P = \set{\wh p}{p \in P}$ the corresponding set of reduced words. 
 Then $\wh P$ is prefix-closed. 
 \end{lemma}

 \begin{proof}
 Let $p \in P$ and $\wh p \in \wh P$ its reduced form. 
 We have to show that every prefix of $\wh p$ belongs to $\wh P$.
 For $p=  1$ this is trivial. Hence, let $p = qa$ with $a\in A$ and $\wh q$ the reduced form of $q$. We have $q \in P$ and, 
 by induction, every prefix of $\wh q$ belongs to $\wh P$.
 Now, if $\wh p$ is a prefix of $\wh q$, we are done. In the other case we have
 $\wh p = \wh q a$. Since $\wh q, \wh p \in \wh P$ we are done  again. 
 \end{proof}

Let us finish to prove that solving equations in idempotent variables over the
free inverse monoid belongs to  $\DEXPTIME$.
 
The input in \prref{thm:marie} is given by a system $\cS$ of equations in idempotent variables over 
a free inverse monoid $\FIM{A}$. 
 Every equation $(U,V)\in \cS$ can be written as in \prref{eq:start}. That is:
 \begin{equation}\label{eq:2start}
u_0 X_1 u_1  \cdots X_{g}u_g = 
v_{0} Y_{1} v_1 \cdots Y_{d}v_d. 
\end{equation} 
In linear time we check that all equations
$u_0 \cdots u_g = 
v_0 \cdots v_d$ hold in the group $\FA$. 
 If one of these equalities is  violated then $\cS$ is not solvable and we can stop.

By Lemma \ref{lem:from_idempotents_to_subsets} it suffices to give a $\DEXPTIME$ algorithm 
for solving  systems of language equations over $\FA$ of the form (\ref{eq:ssfrg}) where the solutions $X_i$'s and $Y_j$'s are required to be nonempty, finite, prefix-closed subsets of $\FA$. 
Thus we may assume that we start with a system 
$\cS$ where every equation can be written as 
\begin{equation}\label{eq:helgaa}
L + \sum_{i\in I} u_iX_i = K + \sum_{j\in J} u_jY_j,
\end{equation}
where  $u_i \in L$ and $u_j \in K$  and $L\cup K$ consist of reduced words, only. 
We say that a solution $\sig: \OO \to 2^{A^*}$ is \emph{strong} if
 $\sig(X)$ consists of reduced words, only. That is $\sig(X) = \pi(\set{u\in A^*}{u \in \sig(X)})$. 
Clearly, $\cS$ has a  solution in $\FA$ \IFF it has a strong solution. 
 
 Next, we transform in deterministic polynomial time the system $\cS$ into a system $\cS_0$ where the equations have a  simple syntactic form. 
 We begin by introducing a fresh variable 
 $X_0$ and an equation $X_0=1$.
  In a second phase, we replace each equation  of type as in (\ref{eq:helgaa})
by two equations using a fresh variable $X_E$ and, since each $u_k \in L_K =\pref(L_K)$ as well as $X_0=1$,  we may define these  equations as follows:\begin{align*}\label{eq:hella}
X_E &= \sum_{u\in L}( uX_0 + \pref(u)) + \sum_{i\in I}( u_iX_i + \pref(u_i)), \\  
X_E&= \sum_{v\in K}( vX_0 + \pref(v)) + \sum_{j\in J} (u_jX_j + \pref(u_i)).
\end{align*}
Thus, there is an equation of the form  $X_0=1$ and a bunch of equations which have 
the form
\begin{align*}
X &= \sum_{k\in K}( u_kX_k + \pref(u_k)) \text{ with } K \neq \es.
\end{align*}
With the help of polynomially many additional fresh variables, it is now obvious that we can transform $\cS$ (with respect to satisfiability)  into an equivalent system $\cS_0$
 containing only three types of
 equations: 
\begin{enumerate}
\item $X= 1$, 
\item $X= Y+Z$, 
\item $X = uY + \pref(u)$, where $u$ is a reduced word.
\item $X= 1+X$ for all $X$. 
\end{enumerate}
The last type of equations $X=1+X$ makes sure that every solution is in nonempty sets containing the empty word. 
(This allows to ignore the restriction that $\sig(X) \neq \es$.)
Phrased differently, without restriction 
$\cS$ is of the form $\cS_0$ at the very beginning. 
At this point we start a nondeterministic polynomial time reduction. 
This means, if $\cS$ has a solution then at least one outcome of the 
nondeterministic procedure yields a solvable system $\cS'$ of language equations. If none of the possible outcomes is solvable then $\cS$ is not solvable.
During this procedure we are going to mark some equations and this forces us to  define the notion of 
solution for systems with marked equations. A \emph{(strong) solution}  
is defined as  a mapping $\sig$ such that each
$\sig(X)$ is given by a  prefix-closed set of (reduced) words in $A^*$ such that all  equations hold as language equations
over
 $\FA$, but all marked equations hold as language equations
over $A^*$ as well. (Thus, we have   a stronger condition
for marked equations.) 
We can think of an ``evolution'' of language
equations over $\FGRI A$  
to language equations over the free monoid $A^*$, and in the middle 
during the evolution we have a mixture of both interpretations. 

Initially we mark all equations of type  $X= 1$, $X= 1 +X$, and $X= Y+Z$. This is possible because  
we may start with a strong solution in nonempty, prefix-closed and finite sets, if $\cS$ is solvable.  

Now we proceed in rounds until all equations are marked. 
We start a round, if  some of the equations $X = uY + \pref(u)$ is not yet marked. If $u = 1$ is the empty word we simply mark that equation, too. Hence we may assume $u \neq 1$ and we  may write $u=va$ with $a \in A$. 
Nondeterministically we guess whether there exists a strong 
solution $\sig$ such that $\ov a \in \sig(Y)$. 

If our  guess is ``$\ov a \notin \sig(Y)$'', then we mark the equation 
$X = vaY + \pref(va)$. 
If the guess is true then marking is correct because then $vaw$ is reduced for all $w \in \sig(Y)$.
Whether or not $\ov a \notin \sig(Y)$ is true, marking an equation 
never introduces new solutions. Thus, a wrong guess does not transform an unsatisfiable system into a satisfiable one.
Hence, it is enough to consider the other case that the guess is  ``$\ov a \in \sig(Y)$'' for some strong solution $\sig$.
In this case we introduce two fresh variables $Y'$, $Y''$ and a new 
marked equation
\begin{equation*}
Y = Y' + \ov a Y'' + \pref(\ov a).
\end{equation*}
If $\ov a \in \sig(Y)$ is correct then we can extend the strong
solution so that
$\ov a \notin \sig(Y')$. If $\ov a \in \sig(Y)$ is false then, again, this step does not introduce any new solution.

Finally, we replace 
the equation $X = vaY + \pref(va)$ by the following three equations, the first two of them are marked and the variables $X'$, $X''$ are fresh
\begin{align*}
X &= X' +X'' &\text{ (marked), } \\
X' &= vaY' + \pref(va) &\text{ (marked), }\\
X'' &= vY'' + \pref(v).
\end{align*}
If the guess ``$\ov a \in \sig(Y)$'' was correct, then the new system has a strong solution.
If the new system has any solution then the old system has a solution because $X'' = vY'' + \pref(v)$ is  unmarked as long as $v \neq 1$. 
After polynomial many rounds all equations are marked. This defines the new system $\cS'$. 
If  $\cS'$  has a solution $\sig'$ then the restriction of $\sig'$ to the original variables  is also a solution of the original system $\cS$. 
If all our guesses were correct with respect to 
a strong solution $\sig$ of $\cS$ then $\cS'$ has a strong solution $\sig'$ such that 
 $\sig$ is the restriction of $\sig'$  to the original variables.
 Hence,  $\cS$ has a  solution \IFF $\cS'$ has a solution.  

It is therefore enough to  consider the system $\cS'$ of language equations over $A^*$.
All the equations are still of one of the types above.  
Let $\sig$ be any mapping from variables in $\cS'$ to finite languages 
of $A^*$, \ie  $\sig(X)\sse A^* $ denotes an arbitrary finite language for all variables. Then we have 
the following implications. 
\begin{itemize} 
\item $\sig(X))= 1$ implies $\pref(\sig(X))= 1$.
\item $\sig(X)= \sig(Y)+\sig(Z)$ implies $\pref(\sig(X))= \pref(\sig(Y))+\pref(\sig(Z))$. 
\item $\sig(X) = u\sig(Y)  + \pref(u)$ implies $\pref(\sig(X)) = u\pref(\sig(Y)) + \pref(u)$. 
\item $\sig(X) = 1 + \sig(X) \iff 1 \in \sig(X) \iff 1 \in \pref(\sig(X))$.
\end{itemize}
Thus, the system $\cS'$ of language equations over $A^*$ has a solution \IFF $\cS'$ has a language solution in nonempty, finite, and prefix-closed sets.

In order to finish the proof, let us briefly repeat what we have done so far. The input has been a 
system $\cS$ of equations over $\FIM A$ in idempotent variables.
If $\cS$ has a solution then it has a strong solution and making all guesses correct we end up with a system $\cS'$ of language equations over $A^*$ which has a strong solution in finite and  prefix-closed sets.
Conversely, consider some system $\cS'$ which is obtained by the 
nondeterministic choices. (Note that the number of different 
systems $\cS'$ is bounded by a singly exponential function and  
$\DEXPTIME$ is enough time to  calculate a list containing all $\cS'$.)
Assume that $\cS'$ has a solution 
$\sig'$ in finite subsets of $A^*$. Due to the syntactic structure of $\cS'$
there is also a solution $\sig$ in  nonempty and prefix-closed subsets of $A^*$. This is due to the three implications above. 
Using \prref{lem:pretopre}, $\sig$ solves $\cS$ as a system of language equations 
over the group $\FGRI A$ in nonempty and prefix-closed subsets of reduced words. Hence,  $\sig$ solves the original system over the free inverse monoid 
$\FIM A$. Since the square of a singly exponential function is singly exponential, it is enough to apply \prref{thm:bn01}. 
\qed

\subsection{Proof of the lower bound in Theorem~\ref{thm:marie}}
\label{sec:lbmarie}
Throughout this section we work over a two letter alphabet 
$B= \os{\alp, \bet}$ which is embedded in the alphabet
$A = \os{\alp, \ov \alp, \bet, \ov \bet}$ with involution without fixed points. Thus, $\abs A = 4$. 

We show that the problem of solving a system of equations in idempotent 
variables over $\FIM{A}$ is $\DEXPTIME$-hard, even if we 
restrict input to systems with two equations. The first part of this lower bound 
proof is about a surgery on language equations. It is the main ingredient, although free inverse monoids do not appear in that part.

\subsubsection{Surgery: {}from solutions in $B^*$ to solutions in $\FA$}\label{sec:nf}
This section contains a sequence of transformations for language 
equations. 
We say that systems $\cS$ and $\cS'$ of language equations 
are 
\emph{sat-equivalent}, provided $\cS$ is solvable \IFF $\cS'$ is solvable.

We consider the equations 
\begin{equation}\label{eq:sat1}
L + \sum_{i\in I} u_{i}X_i= K + \sum_{j\in J} u_{j}X_j
\end{equation}
with coefficients over $B$ and where variables represent finite subsets
of $B^*$. Clearly, if there is a solution over the free monoid $A^*$, then there is also a solution over $B^*$, because the coefficients are over $B$. 
Hence, for sat-equivalence it is enough to consider solutions in finite subsets of $A^*$.

With the help  of a fresh variable $X_0$ each equation 
as in (\ref{eq:sat1}) can be replaced  by the 
following system:
\begin{align*}
X_0 &= 1,\\
\sum_{u\in L}uX_0 + \sum_{i\in I} u_{i}X_i&= \sum_{v\in K}vX_0 + \sum_{j\in J} u_{j}X_j.
\end{align*}
If a term $uX$ appears in a system with  $\abs u \geq 2$, then we write $u=av$ with  $a \in {B}$ and we introduce 
a fresh variable $[{vX}]$.  We replace $uX$ everywhere by $a[{vX}]$; and we add the  equation $[{vX}] = vX$. We can repeat the process until all terms $uX$ satisfy $\abs u \leq 1$. The transformation produces a \sateq system of quadratic size in the original system.
With the help of more fresh variables we can proceed to 
have the following form 

\begin{align*}
X_0 &= 1, \\
X_{1k} & = a_{2k}X_{2k} + a_{3k}X_{3k} \quad \text{ for } 1\leq k \leq n.
\end{align*}
Here, $n \in \N$,  $a_{ik} \in B^*$ 
have length at most $1$, and $X_{ik} \in \OO$.

Next, it is convenient to allow the sign $\leq$  in addition to 
$=$ in the notation of equations. 
More formally, $L\leq R$ denotes the short hand of the language equation 
$L+R = R$.
Vice versa we can identify $L=R$ with the system 
\begin{align*}
L&\leq R,\\
R  &\leq L.
\end{align*}
For example, by letting $a=1$ the equation $X= bY + cZ$ 
is equivalent to the following system where all equations are written in a uniform way. 
\begin{align*}
aX& \leq bY + cZ,\\
bY& \leq aX + aX,\\
cZ& \leq aX + aX.\\
\end{align*}


The transformations above show that on input $\cS$ we can produce in polynomial time a \sateq system $\cS_1$ which can be written as: 
\begin{align}\label{eq:ggs1nf}
X_0 &= 1, \\
a_{1k}X_{1k} & \leq a_{2k}X_{2k} + a_{3k}X_{3k} \quad \text{ for } 1\leq k \leq n.\label{eq:ggs2nf}
\end{align}
As above, $n \in \N$,  $a_{ik} \in B^*$ 
have length at most $1$, and $X_{ik} \in \OO$. 

The next transformation yields a \sateq system which has a \solu \IFF 
it has a \solu in nonempty and prefix-closed sets.

For this we choose some letter $d \in B$ 
and we transform $\cS_1$ into a system $\cS_2$
as follows. We replace the equation $X_0 = 1$ in (\ref{eq:ggs1nf}) by: 
\begin{align}\label{eq:ggs11nf}
X_0 &= 1 + {d}.
\end{align}
Moreover, we replace each equation of type $a_{1k}X_{1k}  \leq a_{2k}X_{2k} + a_{3k}X_{3k}$ in (\ref{eq:ggs2nf})
by:
\begin{align}\label{eq:ggs21nf}
a_{1k}X_{1k} &\leq  a_{1k} + a_{2k}X_{2k} + a_{3k}X_{3k}\end{align}

\begin{lemma}\label{lem:transfeq}
The systems $\cS_1$ and $\cS_2$ are \sateq. Moreover, 
if  $\cS_2$ has any solution, then it has a \solu in nonempty prefix-closed sets.
\end{lemma}

\begin{proof}
Let $\sig: \OO \to 2_f^{B^*}$ be any solution of $\cS_1$. Then  
$$\sig'(X) = 1 + \pref(\sig(X){d})$$ defines a \solu of $\cS_2$ in nonempty prefix-closed sets. Thus, it is enough to show that if $\cS_2$ is solvable, then 
$\cS_1$ is solvable, too. 

To this end, let $\sig'$ be any solution of $\cS_2$.
Define 
$$ \sig(X) = \set{u \in B^*}{u{d} \in \sig'(X)}.$$

(Note that $\sig(X)$ might be empty.) 
Now, $\sig'(X_0) = \os{1,{d}}$ implies $\sig(X_0) = \os{1}$.
It remains to show that $a\sig'(X) \sse \os {a} +  b\sig'(Y) + c\sig'(Z)$
implies $a\sig(X) \sse b\sig(Y) + c\sig(Z)$. This is straightforward.  
Indeed, let $u\in \sig(X)$, hence $u{d} \in \sig'(X)$. Since $au{d}\neq a$  we must 
have $aud \in b\sig'(Y) + c\sig'(Z)$. By symmetry, we may assume 
$aud \in b\sig'(Y)$. Thus, $aud = bvd$ with $vd \in \sig'(Y)$. This implies 
$v \in \sig(Y)$. Therefore, $au \in b \sig(Y)$. Hence, the result.
\end{proof}

The system $\cS_2$ does not suffice for our purpose. We need 
a system where we can control that all solutions $\sig$ and all variables $X$ satisfy $\sig(X) \sse \os{\alp,\bet}^*$. The crucial observation is as follows: let 
$L_1 \lds L_n \sse B^*$ be finite subsets. Then their union is finite, and so is the factor-closure of their union
$$K = \set{v\in B^*}{\exists u, w\in B^*\,\exists 1 \leq i \leq n : uvw \in  L_i}.$$
Factor-closed languages are prefix and  suffix-closed; and for suffix-closed  languages we can control its alphabet by the following condition 
\begin{align}\label{eq:sufalp} 
K \sse \os 1 \cup \bigcup_{a\in B}aK.
\end{align}
More precisely, for every language $K \sse C^*$ where $B \sse C$ we have that 
$K$ satisfies (\ref{eq:sufalp}) \IFF both, $K$ is suffix-closed and $K\sse B^*$. 

\begin{proposition}\label{prop:sys3}
There is a polynomial time algorithm which produces on an input, which is a system $\cS$ of language equations over $B^*$, an output, which is a system of two language equations $\cS'$ satisfying the following conditions.
\begin{itemize}
\item $\cS$ and $\cS'$ are \sateq.
\item If $\cS'$ has any solution, then it has a \solu in nonempty prefix-closed subsets of $B^*$ and therefore a \solu as a language equation over 
the group $\FA$.
\item If $\cS'$ has  a \solu as system of language equations over 
the group $\FA$, then $\cS$ is solvable.
\item The system $\cS'$ can be written in the following syntactic form
\begin{align}\label{eq:syntcSP1}
L + \sum_{i\in I}u_i X_i
&=  K + \sum_{j\in J}v_j X_j,\\
1 + Z+ \sum_{a \in B}aZ  + \sum_{k \in I \cup J}X_k&=   1 + \sum_{a \in B}aZ\label{eq:syntcSP2}
\end{align}
where  $L,K \in 2_f^{B^*}$ and  $u_i,v_j$ denote nonempty words in $B^*$. Moreover, $Z \neq X_k$ for all $k \in I \cup J$. 
\end{itemize}
\end{proposition}
 
\begin{proof}
We may start with the system $\cS_2$ which satisfies \prref{lem:transfeq}.
Since it is \sateq to $\cS_1$, it is also \sateq to $\cS$. The system $\cS_2$ contains equations $X_0 \leq  1+d$ and $1+d\leq X_0$ for some $d \in B$ and all other equations have the 
form $a_{1k}X_{1k}  \leq a_{2k}X_{2k} + a_{3k}X_{3k}$ where $a_{ik}\in \os 1 \cup B$ and $X_{ik}$ are variables. 
By the procedure described in \prref{rem:bn7.6} we transform  $\cS_2$
into a single equation $\cE'$ which has ``almost'' the syntactic form as required in 
\prref{eq:syntcSP1}, because we have
\begin{equation*}
L + \sum_{i\in I} u_i X_i
\leq   K +  \sum_{j\in J} v_j X_j,
\end{equation*}
If $\mathrm{RHS}$ denotes the right-hand side, then we can replace
$$L + \sum_{i\in I} u_i X_i
\leq  \mathrm{RHS}$$
by 
$$L + \sum_{i\in I} u_i X_i + \mathrm{RHS}
= \mathrm{RHS}$$
Hence, a syntactic form as it is required by (\ref{eq:syntcSP1}). 
Recall that these transformations do not change the set of \solu{s}. 
Without restriction, $u_i\neq 1$ and $v_j\neq 1$ for all $i,j$. Moreover, we may assume that there is a variable $Z\in \OO$ which does not appear in 
(\ref{eq:syntcSP1}). Adding \prref{eq:syntcSP2} defines $\cS'$. The system 
$\cS'$ contains two equations.

If $\cS$ is solvable, then (\ref{eq:syntcSP1}) has a solution $\sig$ 
in nonempty prefix-closed subsets of $B^*$. As $Z$ does not appear 
we may assume $\sig(Z) = \os 1$. Redefining 
$$\sig(Z) = \set{v \in B^*}{\exists u \in B^*\, \exists X \in \OO: uv \in \sig(X)}$$ yields a solution 
in nonempty prefix-closed subsets of $B^*$ of the system $\cS'$. 
Hence, a \solu as a language equation over 
the group $\FA$. 

Finally, let $\sig':\OO \to 2_f^{\FA}$ a \solu of $\cS'$ as a language equation over the group $\FA$.
We claim that $\sig'$ is also a \solu in the free monoid 
$A^*$. If so, then $\cS_2$ has a solution in $A^*$, and this implies that 
$\cS$ is solvable. 

By contradiction, assume that $\sig'$ does not solve $\cS'$ over $A^*$. 
Then there are some $b \in B$, $u\in A^*$,  and $X\in \OO$ with $\ov b u \in \sig'(X) \sse \FA$.
We may choose $b$ and $X$ such that $\abs u $ is maximal. 
\prref{eq:syntcSP2} implies $\ov b u = \pi(av)$ for some $a \in B$ and some reduced word $v \in \sig(Z)$. Since $\ov b \neq a$, this implies $v= \ov a w$ and  $\pi(av) = w$. Hence 
$\ov a \ov b u \in Z$, which contradicts that $u$ was of maximal length.
\end{proof}

\subsubsection{Finishing the proof of Theorem~\ref{thm:marie}}\label{sec:mariemainliftsin}

Due to \prref{thm:bn01} and \prref{prop:sys3} we know that the  problem to decide systems $\cS'$ with two language equations in the form of \prref{prop:sys3} 
is $\DEXPTIME$-complete. Thus, all we need to finish the proof of \prref{thm:marie} is the following lemma.

\begin{lemma}\label{lem:marieade}
There is polynomial time algorithm which produces on an input equation  $\cS'$ 
as in \prref{prop:sys3} a system $\cS''$ of two equations $U_1=V_1$ and $U_1 = V_2$ over $\FIM{A}$ in idempotent variables 
such that $\cS'$ is solvable as a language equation \IFF $\cS''$ is solvable over $\FIM{A}$.
\end{lemma}

\begin{proof}
Consider  the system $\cS'$ in \prref{prop:sys3} and let 
$\text{LHS}_i$ resp.{} $\text{RHS}_i$ be the left- resp.{} right-hand sides of the corresponding equations, $i= 1,2$. 
Each of these terms has the form $T= L + \sum_{i\in I} u_i X_i$ which is 
defines a word $W(T)$ by  
$$W(T) = \prod_{u \in L} u \ov u \cdot \prod_{i \in I}u_i X_i \ov{u_i}.$$
The ordering in the products can be chosen arbitrarily.  We define $\cS''$ by a system of two equations:
\begin{align*}
W(\text{LHS}_1) & = W(\text{RHS}_1), \\ 
W(\text{LHS}_2) & = W(\text{RHS}_2).
\end{align*}
If $\cS'$ is solvable in nonempty prefix-closed subsets of $B^*$, then 
$\cS''$ is solvable in $\FIM{A}$. Conversely, if $\cS''$ is solvable in $\FIM{A}$, then $\cS'$ has  a \solu as system of language equations over 
the group $\FA$. 
\end{proof}

\prref{prop:sys3} and \prref{lem:marieade} conclude the proof of \prref{thm:marie}. 
%

\begin{thebibliography}{10}

\bibitem{Appel68}
K.~I. Appel.
\newblock One-variable equations in free groups.
\newblock {\em Proc. Amer. Math. Soc.}, 19:912--918, 1968.

\bibitem{BaaderN01}
F.~Baader and P.~Narendran.
\newblock Unification of concept terms in description logics.
\newblock {\em J. Symb. Comput.}, 31:277--305, 2001.

\bibitem{BormotovGM09}
D.~Bormotov, R.~Gilman, and A.~Myasnikov.
\newblock Solving one-variable equations in free groups.
\newblock {\em J. Group Theory}, 12:317–--330, 2009.

\bibitem{CiobanuDEicalp2015}
L.~Ciobanu, V.~Diekert, and M.~Elder.
\newblock Solution sets for equations over free groups are {EDT0L} languages.
\newblock In M.~Halld{\'o}rsson, K.~Iwama, N.~Kobayashi, and B.~Speckmann,
  editors, {\em Proc. 42nd International Colloquium Automata, Languages and
  Programming (ICALP 2015), Part~{II}, Kyoto, Japan, July 6-10, 2015}, volume
  9135 of {\em Lecture Notes in Computer Science}, pages 134--145. Springer,
  2015.

\bibitem{DMS07ijac}
T.~Deis, J.~C. Meakin, and G.~S{\'e}nizergues.
\newblock Equations in free inverse monoids.
\newblock {\em IJAC}, 17:761--795, 2007.

\bibitem{DiekertJP2014csr}
V.~Diekert, A.~Je\.z, and W.~Plandowski.
\newblock Finding all solutions of equations in free groups and monoids with
  involution.
\newblock In E.~A. Hirsch, S.~O. Kuznetsov, J.~Pin, and N.~K. Vereshchagin,
  editors, {\em Computer Science Symposium in Russia 2014, {CSR} 2014, Moscow,
  Russia, June 7-11, 2014. Proceedings}, volume 8476 of {\em Lecture Notes in
  Computer Science}, pages 1--15. Springer, 2014.

\bibitem{DiekertMSScsr15}
V.~Diekert, F.~Martin, G.~S{\'{e}}nizergues, and P.~V. Silva.
\newblock Equations over free inverse monoids with idempotent variables.
\newblock In L.~D. Beklemishev and D.~V. Musatov, editors, {\em Proc. 10th
  International Computer Science Symposium in Russia, CSR 2015, Listvyanka,
  Russia, July 13-17, 2015}, volume 9139 of {\em Lecture Notes in Computer
  Science}, pages 173--188. Springer, 2015.

\bibitem{gut2000stoc}
C.~Guti{\'e}rrez.
\newblock Satisfiability of equations in free groups is in {PSPACE}.
\newblock In {\em Proceedings 32nd Annual ACM Symposium on Theory of Computing,
  STOC'2000}, pages 21--27. ACM Press, 2000.

\bibitem{Jez15jacm}
A.~Je\.z.
\newblock Recompression: a simple and powerful technique for word equations.
\newblock {\em J. ACM}, 2015.
\newblock To appear. The conference version is in the Proc.~STACS 2013 :LIPIcs
  {\bf{20}}, 233--244 (2013). Schloss Dagstuhl--Leibniz-Zentrum f{\"u}r
  Informatik.

\bibitem{Lor68}
A.~A. Lorents.
\newblock Representations of sets of solutions of systems of equations with one
  unknown in a free group.
\newblock {\em Dokl. Akad. Nauk.}, 178:290--292, 1968.
\newblock (in Russian).

\bibitem{lot83}
M.~Lothaire.
\newblock {\em Combinatorics on Words}, volume~17 of {\em Encyclopedia of
  Mathematics and its Applications}.
\newblock Addison-Wesley, Reading, MA, 1983.
\newblock Reprinted by {\em Cambridge University Press}, 1997.

\bibitem{mak77}
G.~S. Makanin.
\newblock The problem of solvability of equations in a free semigroup.
\newblock {\em Math. Sbornik}, 103:147--236, 1977.
\newblock English transl. in Math. USSR Sbornik 32 (1977).

\bibitem{mak82}
G.~S. Makanin.
\newblock Equations in a free group.
\newblock {\em Izv. Akad. Nauk SSR}, Ser. Math. 46:1199--1273, 1983.
\newblock English transl. in Math. USSR Izv. 21 (1983).

\bibitem{Munn:74}
W.~D. Munn.
\newblock {Free inverse semigroups}.
\newblock {\em Proc. London Math. Soc.}, 29:385--404, 1974.

\bibitem{pap94}
{\Ch}.~H. Papadimitriou.
\newblock {\em Computational Complexity}.
\newblock Addison Wesley, 1994.

\bibitem{pet84}
M.~Petrich.
\newblock {\em Inverse semigroups}.
\newblock Wiley, 1984.

\bibitem{pla99focs}
W.~Plandowski.
\newblock Satisfiability of word equations with constants is in {PSPACE}.
\newblock In {\em Proc. 40th Ann.~Symp. on Foundations of Computer Science,
  FOCS'99}, pages 495--500. IEEE Computer Society Press, 1999.

\bibitem{pla04jacm}
W.~Plandowski.
\newblock Satisfiability of word equations with constants is in {PSPACE}.
\newblock {\em J. ACM}, 51:483--496, 2004.

\bibitem{roz85}
B.~V. Rozenblat.
\newblock Diophantine theories of free inverse semigroups.
\newblock {\em Siberian Math. J.}, 26:860--865, 1985.
\newblock Translation from Sibirskii Mat. Zhurnal, volume 26: 101--107, 1985.

\bibitem{Scheiblich73}
H.~E. Scheiblich.
\newblock Free inverse semigroups.
\newblock {\em Proc. Amer. Math. Soc.}, 38:1--7, 1973.

\bibitem{Silva99}
P.~V. Silva.
\newblock Word equations and inverse monoid presentations.
\newblock In S.~Kublanovsky, A.~Mikhalev, P.~Higgins, and J.~Ponizovskii,
  editors, {\em Semigroups and Applications, Including Semigroup Rings}.
  Severny Ochag, St. Petersburg, 1999.

\end{thebibliography}
\section{One-variable equations}\label{sec:onevar}
Throughout this section we assume that the involution 
on $A$ is without fixed points, \ie $\FGRI A$ is equal to the free group $\FG (A_+) $ in the standard terminology. It is open whether we can remove this restriction. 

The following notation is defined for any alphabet $\Sig$  and any nonempty word $p \in \Sig^+$. 
For $u \in \Sig^*$  we let ${\abs u}_p$ be the number of occurrences of $p$ as a factor in $u$. Formally:
\begin{equation*}\label{eq:fred}
{\abs u}_p = \abs{\set{u'}{u'p \leq u}}.
\end{equation*}
The following equation 
is trivial since $p$ may occur across the border 
between  $u$ and $v$ at most  $\abs p-1$ times. 
\begin{equation}\label{eq:freddy}
0 \leq {\abs {uv}}_p - {\abs u}_p - {\abs v}_p \leq  \abs p -1.
\end{equation}
Next, assuming that $\Sig$ is equipped with an involution, we define a ``difference'' function $\del_p: \Sig^* \to \Z$ by 
\begin{equation*}\label{eq:otto}
\dip u p  = {\abs u}_p - {\abs u}_{\ov p}.
\end{equation*}
Since $\dip u p =  \dip {\ov u} {\ov p}$ we have $\dip u p = - \dip {\ov u} p$, and the mapping $\del_p$ respects the involution. 

By definition, we have
\[\delta_p(uv) -\delta_p(u)-\delta_p(v)  = ( |uv|_p -|u|_p-|v|_p) - ( |uv|_{\overline{p}} -|u|_{\overline{p}} - |v|_{\overline{p}})\]
Hence, we may use \prref{eq:freddy} to conclude:
\begin{equation}\label{eq:frod}
\abs {\dip {uv} p  - \dip u p  - \dip v p} \leq  \abs p -1.
\end{equation}

As we identify $\FGRI \Sig$ with the subset of reduced words in $\Sig^*$, the mapping $\del_p$ is defined from $\FGRI \Sig$ to $\Z$, too. 
The next lemma shows that its deviation from being a \homo can be upper  bounded. 
The next lemma will be applied to a primitive word $p$, only. Let us remind that a  word is defined to be primitive if it  cannot be written in the form $v^i$ for some word $v$ with $i>1$ and it is not empty .
Every nonempty word $u$ has a {\em primitive root}: it is the uniquely defined  primitive word $p$ such that $u \in p^+$. 
\begin{lemma}
\label{lem:bapra}
Let $u_1, \ldots,u_n, p $ be reduced words with $p \neq 1$. 
Let $w$ be the uniquely defined reduced word such that 
$w$ is equal to $u_1 \cdots u_n$ in the group $\FGRI \Sig$. 
Then we have: 
\beq
\label{eq:bapra2}
\abs{\dip {w} p 
- \dip {u_1}{p}  - \cdots  - \dip {u_n} p} \leq
3(\abs{p} -1)(n-1).
\eeq
\end{lemma}

\begin{proof}
 Clearly, \prref{eq:bapra2} holds for $n = 1$.
Hence, let $n \geq 2$. Let $u$ be the reduced word such that 
$u_1 \cdots u_{n-1}$ reduces to $u$.
By induction, we have $\abs{\dip {u} p 
- \dip {u_1}{p}  - \cdots  - \dip {u_{n-1}} p} \leq
3(\abs{p} -1)(n-2)$.
Let $v = u_n$. By triangle inequality it is enough to show 
\beq
\label{eq:bapra3}
\abs{\dip {w} p 
- \dip {u}{p} -  \dip {v} p} \leq
3(\abs{p} -1).
\eeq
To see this write $u = u'r$ and $v = \ov r v'$ such that 
$w = u'v'$. 
\begin{align*}
\dip {w} p - \dip {u}{p} -  \dip {v} p &= \dip {w} p - \dip {u'}{p} -  \dip {v'} p\\
&+ \dip {u'}{p} + \dip {r}{p} - \dip {u}{p} \\
&+ \dip {\ov r}{p}+\dip {v'}{p}  - \dip {v}{p}
\end{align*}
The result follows by \prref{eq:frod} and triangle inequality. 
\end{proof}

We will apply \prref{lem:bapra} in the following equivalent form.
\begin{multline}
\label{eq:baprpr}
 \dip {u_1}{p}  + \cdots  + \dip {u_n} p - 3(\abs{p} -1)(n-1)   \leq  \dip {w} p   \\
   \leq \dip {u_1}{p}  +  \cdots  +  \dip {u_n}p +  3(\abs{p} -1)(n-1).
\end{multline}
The following lemma is easy to prove. It is however here where we use  $a \neq \ov a$ for all $a \in A$. Let us recall that a word $q$ is  \emph{cyclically reduced} if  $qq$ is reduced. In other words if $a$ is the first letter of $q$, the last letter of $q$ is different from $\overline{a}$.

\begin{lemma}
\label{lem:pricr}
Let $n \in \Z$ and $q \in {\FA}$ be a primitive and cyclically reduced word. Then we have $\del_q(q^n)
= n.$
\end{lemma}

\proof
We may assume without loss of generality that $n > 0$. Clearly,
$|q^n|_q \geq n$. Suppose that $|q^n|_q > n$. Then $q$ is a proper
factor of $qq$, hence we may write $q =
q_1q_2 = q_2q_1$ in reduced products with $q_1,q_2 \neq 1$. It is well
known (see e.g. \cite{lot83}) that this contradicts the primitivity of
$q$. Thus, $|q^n|_q = n$.

Suppose now that $\ov q$ is a proper
factor of $qq$. Then we may write $q =
q_1q_2$ as a reduced product with $\ov q= q_2q_1$ since $q$ is cyclically reduced. Moreover, since $\ov q  =
\ov q_2\ov q_1$ we get $q_2 = \ov q_2$ and $q_1 = \ov q_1$. Hence
$q_1 = q_2 = 1$ because $q_1$, $q_2$ are reduced and $a \neq \ov a$ for all $a \in A$.
Thus, $|q^n|_{\ov q} = 0$ and so $\del_q(q^n)
= |q^n|_{q} - |q^n|_{\ov q}  = |q^n|_{q}= n$.
\qed

An (untyped) equation $(U,V)$ is called a 
\emph{one-variable equation}, 
 if we can write $UV \in (A \cup  \os{X , \ov X})^*$. 
 More generally,  
 we also consider systems of typed equations 
with at most one  reduced variable $x$ (and $\ov x$), \ie
every equation $(U,V)$ in the system satisfies 
$UV \in (A \cup \OO \cup  \os{x, \ov x})^*$. 
  Let us fix some more notation, we let
$\Sig= A \cup \OO \cup \Gam$ with $\Gam= \os{x , \ov x}$.
In particular, we have $\ov X = X$ for all $X \in \OO$ and 
$\alpha \neq \ov \alpha$ for all $\alpha \in A \cup \Gam$.

\begin{definition}\label{def:unbw}
Let $u,v \in \Gam^*$.  We say that $(u,v)$ is
\emph{unbalanced} if $u \neq v$ in the free inverse monoid 
$\FIM{\Gam}$. 

Otherwise we say that $(u,v)$ is \emph{balanced}. 
\end{definition}

\begin{remark}\label{rem:remsein}
Using the well-known structure of $\FIM{\Gamma}$,
 a pair $(u,v)$ 
as in \prref{def:unbw} is balanced \IFF the following three conditions
are satisfied.\\ 
\begin{itemize}
\item $\dip u x  = \dip v x$.
\item 
$\max\{ \delta_x(u') \mid u' \leq u\} = \max\{ \delta_x(v') \mid v' \leq v\}$.
\item 
$\min\{ \delta_x(u') \mid u' \leq u\} = \min\{ \delta_x(v') \mid v'
\leq v\}$.
\end{itemize}
\end{remark}

We extend the notion defined in \prref{def:unbw} to an untyped one-variable equation. In the following we let 
$\pi_{A,\Gam}$ be the morphism from $(A \cup \OO \cup \Gam)^*$ 
to $\FGRI {A \cup \Gam}$ which is induced by cancelling the symbols 
in $\OO$. 
\begin{definition}\label{def:numberedsunny}
Let  $(U,V)$ be an untyped one-variable equation
with $\XX = \os{X , \ov X}$.
We say that $(U,V)$ is  \emph{unbalanced}
if it fulfills both conditions:\\
1- $(u,v)$ is unbalanced as a word over $\Gam$ where $u$ (resp. $v$) is obtained
from $U$ (resp. $V$) by replacing $X$ by $x$ (and $\ov X$ by $\ov x$) and erasing
all other symbols. \\
2- $\pi_{A,\Gam}(U) \neq \pi_{A,\Gam} (V)$ in the free group $\FGRI{A \cup \Gam}$. 
\end{definition}
The following definition is a bit more technical, but it will lead to better  results. 
\begin{definition}\label{def:sun}
Let $U,V$ be words over $A\cup \OO \cup \Gam$.
 We say that $(U,V)$ is  \emph{strongly unbalanced } 
if $\pi_{A,\Gam}(U) \neq \pi_{A,\Gam} (V)$ in the free group $\FGRI{A
  \cup \Gam}$ and at least one of the following conditions is satisfied.
\begin{itemize}
\item[(SU1)] $\dip U x  \neq  \dip V x$.
\item[(SU2)] For all ${z} \in \OO \cup \os 1$ and all prefixes 
$V'{z}$ of $V$ there exists some prefix $U'{z}$ of $U$
such that $\dip {U'} { x}  > \dip {V'} {x}$.
\item[(SU3)] For all ${z} \in \OO \cup \os 1$ and all prefixes 
$V'{z}$ of $V$ there exists some prefix $U'{z}$ of $U$
such that $\dip {U'} {\ov x}  > \dip {V'} {\ov x}$.
\end{itemize}
\end{definition}

The following result improves the complexity in the corresponding statement of
\cite{DMS07ijac}. (Note that the condition $\pi_{A,\Gam}(U) \neq \pi_{A,\Gam} (V)$ 
was missing in \cite{DMS07ijac}, but the proof is not valid without this additional requirement.) 
\begin{theorem}\label{thm:lea}
The following problem can be decided in \DEXPTIME.

Input: A system $\cS$ of one-variable equations over $\XX = \os{X , \ov X}$ where at least one equation  $(U,V)$ is unbalanced according to  
\prref{def:numberedsunny}.

Question: Does $\cS$ have a solution in $\FIM A$?
\end{theorem}

\begin{proof}
Suppose that $(U,V)$ is unbalanced. The pair $(U,V)$ must then contradict one of the three conditions of Remark \ref{rem:remsein}. Let us distinguish cases 
and, in each case, reduce the given unbalanced equation  into a 
{\em strongly} unbalanced {\em typed} equation.\\
In all cases, we introduce a fresh idempotent variable $Z$, a fresh reduced variable $x$, and use the word-substitutions $\tau'$ (or $\tau $) defined in Lemma \ref{lem:taueqtoreq}: $\tau'(X)=xZ,\tau'(\ov{X})= \ov{Z}\ov{x},\tau(X)=Zx,\tau(\ov{X})= \ov{x}\ov{Z}$
or the trivial substitution  $\theta(X)=x,\theta(\ov{X})= \ov{x}$.\\
{\bf Case 1}: $\delta_X(U) \neq \delta_X(V)$.\\
In this case $(\theta(U),\theta(V))$ fulfills condition $(SU1)$.\\
{\bf Case 2}: $\max\{ \delta_X(U') \mid U' \leq U\} > \max\{ \delta_X(V') \mid
  V' \leq V \}$.\\
There is some prefix 
 $U' \leq U$ such that for all prefixes $V' \leq V$ we have
 $\delta_X(U') > \delta_X(V')$ and,
in particular,  $\delta_X(U') >
  \delta_X(1) = 0$. 
  We choose $\delta_X(U')$ to be maximal
and, since $\delta_X(U')$  is positive, we may
choose $U'$ such that  $X = \last(U')$, 
so that   $\last(\tau'(U'))=Z$. 
Now,  for every $z \in \{Z,1\}$,
\begin{eqnarray*}
\dip {\tau'(U')} {x} = \delta_X(U') &>& \max\set{\delta_X(V')}{V' \leq V}\\
&= & \max\set{\dip {W} {x}}{W \leq \tau'(V)}\\
& \geq &\max\set{\dip {W'z} {x}}{W'z \leq \tau'(V)}.
\end{eqnarray*}
This prefix $\tau'(U')$ shows that $(\tau'(U),\tau'(V))$ fulfills condition $(SU2)$
(this is actually a stronger 
 requirement than asked by \prref{def:sun}, because this  single
 prefix $\tau'(U')$ serves for all $W'z$). \\
{\bf Case 2'}: $\max\{ \delta_X(U') \mid U' \leq U\} < \max\{ \delta_X(V') \mid
  V' \leq V \}$.\\
By Case 2 the typed equation $(\tau'(V),\tau'(U))$ fulfills condition $(SU2)$.\\
{\bf Case 3}: $\min\{ \delta_X(U') \mid U' \leq U\} > \min\{ \delta_X(V') \mid
  V' \leq V \}$.\\
We may assume that $\delta_X(U) = \delta_X(V) = k$. If $U = U'U''$ and
$V = V'V''$, we have 
$\delta_{\ov{X}}(\ov{U''}) = \delta_X(U'') = k-\delta_X(U')$ and
$\delta_{\ov{X}}(\ov{V''}) = \delta_X(V'') = k-\delta_X(V')$, thus
$(\ov{U},\ov{V})$ fulfills that
$\max\set{\dip{U'}{\ov X }}{U' \leq \ov{U}} < 
\max\set{\dip{V'}{\ov X}}{V' \leq \ov{V}}$.\\
By a reasoning similar to that of case 2, one can show that
$(\tau(\ov{V}),\tau(\ov{U}))$
fulfills condition $(SU3)$.\\
{\bf Case 3'}: $\min\{ \delta_X(U') \mid U' \leq U\} < \min\{ \delta_X(V') \mid
  V' \leq V \}$.\\
By Case 3 the typed equation 
$(\tau(\ov{U}),\tau(\ov{V}))$ 
fulfills condition 
$(SU3)$.\\
We have thus reduced \prref{thm:lea} above to  
\prref{thm:frida} below. 
\end{proof}

\begin{theorem}\label{thm:frida}
The following problem can be decided in \DEXPTIME.

Input: A system $\cS$ of typed equations 
with at most one reduced variable (\ie $\Gam = \os{x , \ov x}$) where at least one equation $(U,V)\in \cS$ is strongly unbalanced.  

Question: Does $\cS$ have a solution in $\FIM A$?
\end{theorem}

The proof of \prref{thm:frida} relies on the following combinatorial observation. 

\begin{lemma}\label{lem:gunnar}
Let $(U,V)$ be a  strongly unbalanced equation with 
$U,V \in (A\cup \OO \cup \os{x , \ov x})^*$ and $n = \max\os{\abs{U}, \abs{V}}$.
Let $k \in \Z$ be an integer and  $\sig$ be a solution to $(U,V)$ such that $\sig(x) 
= (\pref(p^k), p^k)$ for some nonempty cyclically reduced word $p\in A^*$. Then we have $\abs k \leq 6 n \abs p$.
\end{lemma}

\begin{proof}
 Without restriction,  $p$ is primitive and $k >1$. (Replace $p$ by its primitive root and interchange the role of $p$ and $\ov p$, if necessary.) For a word $W \in  (A\cup \OO \cup \os{x , \ov x})^*$ 
 we write $\sig(W) = (\sig_1(W), \sig_2(W))$ where 
 $\sig_1(W) \sse A^*$ is a prefix-closed set of reduced words
  and $\sig_2(W) \in \FA$.
 Choose $m$ maximal such that 
 $\dip w p =m $ for some $w \in \sig_1(V)$. We fix $w \in A^*$ and we observe
 that we have $m \geq 0$ and for every word 
 $u \in  \sig_1(U) = \sig_1(V)$, we have
\begin{equation}\label{eq:m_is_maximum}
 \dip u p \leq m 
\end{equation} 
{\bf Case (SU2):} $(U,V)$ fulfills condition (SU2).\\
We choose  a prefix $V'$ of $V$ of minimal length with respect to
the property $w \in \sig_1(V')$. 
We consider two subcases.
\smallskip

\noindent
\underline{Subcase $\OO$:} $\last(V') \in \OO$.

\smallskip

\noindent
Let $z:= \last(V')$.
The word $V'$ thus decomposes as $V' = V''z$.
Since $\sig_1(V''{z}) =
 \sig_1(V'') \cup \sig_2(V'')\sig_1(z)$, it follows from the minimality
 of $V''z$ that $w \in \sig_2(V'')\sig_1(z)$. Since $\sig_2(Z) = 1$
 for every $Z \in \OO$ and $|V'| \leq n-1$, it follows that $w$ is the
 product of at most $n-1$ reduced words $v_1\ldots
 v_t$ in $A \cup \{ \sig_2(x), \sig_2(\ov x) \}$ by some $z'\in \sigma_1(z).$\\

For each letter $a$ of $A$,  $\delta_p(a)\leq 1$ and, since $p$ is primitive, by \prref{lem:pricr},
$\delta_p(\sig_2(x)) = k, \delta_p(\sig_2(\ov x)) = -k$. We thus  get
\begin{equation}\label{eq:sumdelta_upperbound}
\sum_{i=1}^t \delta_p(v_i) \leq k\delta_x(V') +n-1.
\end{equation}
Since $w = v_1\ldots v_tz'$, we obtain the following upper bound: 
\begin{eqnarray}\label{eq:m_upperbound}
m = &\dip w p& \nonumber\\
\leq & \sum_{i=1}^t \delta_p(v_i)+ \dip {z'} p +  3 (\abs p -1)(n-1) &\mbox { by } (\ref{eq:baprpr})\nonumber\\
\leq &   k \dip {V'} x + \dip {z'} p + n -1 + 3 (\abs p -1)(n-1)&\mbox { by } (\ref{eq:sumdelta_upperbound})
\end{eqnarray}

By (SU2) there exists a prefix $U'$ of $U$
such that $\dip {U'} x > \dip {V'} x$ and $\last(U')=z$. The word $U'$ thus decomposes as $U' = U''z$.
Let us define $u := \sig_2(U'')z'$. We remark that $u \in \sig_1(U)$, hence it fulfills 
\prref{eq:m_is_maximum}.
Using similar arguments based on \prref{eq:baprpr} and \prref{lem:pricr} we obtain:
\begin{equation}\label{eq:m_lowerbound}
 k \dip {U'}  x + \dip {z'} p - (n-1) - 3 (\abs p -1)(n-1)\leq \dip u p  .
\end{equation}
Combining the above inequalities we obtain:
\begin{eqnarray*}\label{eq:k_upperbound}
k \leq & k(\dip {U'}  x - \dip {V'}  x) &\mbox{\hspace{-2cm} since } \dip {U'}  x > \dip {V'}  x \nonumber\\
 \leq & -\dip {z'} p + (n-1) + 3 (\abs p -1)(n-1)+ \dip u p- k\dip {V'}  x &\mbox{ by } 
(\ref{eq:m_lowerbound}) \nonumber\\
 \leq & -\dip {z'} p + (n-1) + 3 (\abs p -1)(n-1)+ m - k\dip {V'}  x &\mbox{ by } (\ref{eq:m_is_maximum})\nonumber\\
 \leq & 2(n-1) + 6 (\abs p -1)(n-1)&\mbox{ by } (\ref{eq:m_upperbound})\nonumber\\
 \leq & 6 n(\abs p)&
\end{eqnarray*}
\smallskip

\noindent
\underline{Subcase 1:} $\last(V') \notin \OO$.

\smallskip

\noindent
We just need to perform some adaptations to the preceding case. 
The word $w$ is the
 product of at most $n$ reduced words $v_1\ldots
 v_t$ in $A \cup \{ \sig_2(x), \sig_2(\ov x) \}$, 
and by similar methods we obtain
\begin{equation}\label{eq:newm_upperbound}
m = \dip w p \leq  k \dip {V'} x + n + 3 (\abs p -1)(n-1).
\end{equation}

By (SU2) (where we choose $z := 1$) there exists a prefix $U'$ of $U$
such that $\dip {U'} x > \dip {V'} x$. Let us define $u := \sig_2(U')$. We get
\begin{equation}\label{eq:newm_lowerbound}
 k \dip {U'}  x  - n - 3 (\abs p -1)(n-1)\leq \dip u p  .
\end{equation}
Since $u=\sig_2(U') \in \sigma_1(U)$, here also $u$ fulfills (\ref{eq:m_is_maximum}). Hence, 
putting (\ref{eq:newm_lowerbound}) (\ref{eq:m_is_maximum}) and (\ref{eq:newm_upperbound}) together we obtain the desired result: 
\begin{equation}\label{eq:newk_upper_bound}
k \leq k(\dip {U'}  x - \dip {V'}  x) \leq 6 (\abs p - 1)(n-1) + 2n \leq 
6 n \abs p . 
\end{equation}

\noindent {\bf Case (SU3):} $(U,V)$ fulfills condition (SU3).\\
This case is dealt with in a similar manner.

\noindent {\bf Case (SU1):} $(U,V)$ fulfills condition (SU1).\\
By symmetry in $U$ and $V$,  we may assume 
without restriction $\dip U x  >  \dip V x$.
Let us  choose $V' := V,w := \sigma_2(V), m := \dip w p, U' := U, u := \sigma_2(U)$. 
The arguments of Case (SU2), Subcase 1, apply on these choices for
$V',w, m, U', u $.\\ 
(In fact, an argument provided by James Howie in
 \cite{Silva99} shows that in this case the solution $\sigma(x)$ is unique).
\end{proof}

{\noindent \bf Proof of \prref{thm:frida}.}
 Let $n$ be the size of the system $\cS$, it is defined as
 $$\Abs \cS = \sum_{(U,V) \in \cS}\abs {UV}.$$
Since $\pi_{A,\Gam}(U) \neq \pi_{A,\Gam}(V)$ for at least one equation in the system, the set of solutions for the underlying group equations is never equal to $\FGRI A$. 
 By \cite{Appel68,Lor68}, the set of
solutions of a one-variable free group equation is therefore a finite union of sets of the form
\beq
\label{boupow7}
\set{ rq^ks }{ k \in \Z },
\eeq
where $q$ is cyclically reduced and both products
$rqs$ and $r\ov q s$ are reduced. A self-contained proof
of this fact has been given in \cite{BormotovGM09}. 

In the description above  $q=1$ is possible. Moreover, \cite{BormotovGM09} shows
$\abs{rqs} \in \Oh(n)$. Hence, as we aim for \DEXPTIME{} there is time enough to consider all possible 
candidates for $r$ and $s$. This means we can fix $r$ and $s$; and it is enough to  consider a single set $S= \set{ rq^ks }{ k \in \Z }$, only. Next we replace in $\cS$ all occurrences of $x$ by 
$r xs$ (and $\ov x$ by 
$\ov s \, \ov x\,  \ov r$). This leads to a new system 
which we still denote by $\cS$ and without restriction we have $S= \set{q^k}{ k \in \Z }$. The new size $m$ of $\cS$ is at most quadratic in $n$.

Now, we check 
if $k =0$ leads to a solution of $\cS$. This means that we simply 
cancel $x$ and $\ov x$ everywhere. We obtain a system over idempotent 
variables; and we can check satisfiability by \prref{thm:marie}. 
Note that this includes the case $q=1$. Thus, henceforth we may assume that $q$ is a primitive cyclically reduced word. By \prref{lem:gunnar} 
we see that it is enough to replace $S$ by
$S'= \set{q^k}{ \abs k \leq 6 m \abs q }$.  
Since $\abs q \in \Oh(m)$ we obtain a cubic bound for the maximal 
length of words in $S'$, this means the length of each word in 
$S'$ is bounded by $\Oh(n^6)$. This is small enough to check 
satisfiability of the original system $\cS$  in \DEXPTIME{} by \prref{thm:marie}.
\qed

\subsection*{Conclusion and directions for future research}
The notion of ``idempotent variable'' unifies the approach for studying equations 
in free inverse monoids. As the general situation is undecidable, progress is possible only by improving complexities in classes where decidability is known and/or to enlarge the class of equations where decidability is possible. We achieved progress in both fields. For equations in idempotent variables we lowered the complexity down to 
$\DEXPTIME$ and proved that this bound is tight.
Using a recent result in \cite{DiekertJP2014csr} that it is decidable 
in $\PSPACE$ whether an equation in free groups has only finitely many solutions, 
we derived  a ``promise result'' in \prref{cor:mainsota} with  triple exponential 
time complexity. We don't think that this is optimal, because we believe that
solving equations in free groups is in $\NP$. But this fundamental  conjecture is wide open and resisted all known techniques. 

More concretely, let us resume some interesting and specific problems on equations in free inverse
monoids which are open:
\begin{itemize}
\item Is the decision problem in \prref{thm:marie} restricted to  single equation  in idempotent variables \DEXPTIME-hard? We conjecture: yes! 
\item Is the (other) special kind of equations solved by Theorem~23 of \cite{DMS07ijac}
also solvable in \DEXPTIME?
\item Is it possible to remove Assumption~2 in \prref{def:numberedsunny},  and still
maintain decidability of the system of equations? (The assumption asserts that the image of the left-hand side and right-hand side are different in the free group.)
 \item What happens if the underlying equation in the free group is true for all 
elements in the free group? This means the statement is a tautology  the free group.
\item What {\em more general} kinds of one-variable equations in the free inverse
monoid are algorithmically solvable (possibly all of them)?
\item Does Je\.z' recompression technique apply to language equations? If yes, then this would open a new approach to tackle equations over free inverse monoids. 
\end{itemize}

\subsection*{Acknowledgements}
Florent Martin acknowledges support from Labex CEMPI (ANR-11-LABX-0007-01) and SFB 1085 Higher invariants. 
Pedro Silva acknowledges support from:
CNPq (Brazil) through a BJT-A grant (process 313768/2013-7);
and
the European Regional
Development Fund through the programme COMPETE
and the Portuguese Government through FCT (Funda\c c\~ao para a Ci\^encia e a Tecnologia) under the project PEst-C/MAT/UI0144/2013.
Volker Diekert thanks the hospitality of Universidade Federal da Bahia, Salvador Brazil, where 
part of this work started in Spring 2014. 

The authors are thankful to the program committee of CSR 2015 for awarding the conference version of this paper with a Yandex-best-paper award; and one of the authors is even more thankful for the memorable event of Computer Science in Russia 2015 which was held at the shores of a truly magnificent Lake Baikal.

\newcommand{\Ch}{Ch}\newcommand{\Yu}{Yu}\newcommand{\Zh}{Zh}\newcommand{\St}{St}
\newcommand{\curlybraces}[1]{\{#1\}}

\end{document}